%% file: Main.tex
\documentclass[11pt, final, onecolumn]{article}
\pdfoutput=1

\usepackage[backend=bibtex,style=alphabetic,maxbibnames=99,maxalphanames=3,natbib=true]{biblatex}
\usepackage{amsmath, amsfonts, amssymb, amsthm, graphicx, enumitem, titlesec, comment, hyperref}
\usepackage[font=footnotesize,labelfont=bf]{caption}
\allowdisplaybreaks
\titlelabel{\thetitle.\quad}
\usepackage[margin=1in]{geometry}

\newcommand{\vect}[1]{\mathbf{#1}}

\newcommand{\RR}{\mathbb{R}}

\newcommand{\FF}{\mathbb{F}}
\newcommand{\PP}{\mathbb{P}}
\newcommand\norm[1]{\left\lVert#1\right\rVert}
\newcommand\abs[1]{\left\lvert#1\right\rvert}

\newcommand\parens[1]{\left(#1\right)}

\newcommand\angles[1]{\langle#1\rangle}
\newcommand*{\pr}[2][]{\PP\ifx\\\left[#1\right]\\\else_{#1}\fi \left[#2\right]}
\newcommand*{\EE}[2][]{\mathbb{E}\ifx\\\left[#1\right]\\\else_{#1}\fi \left[#2\right]}

\newtheorem{defin}{Definition}[section]
\newtheorem{lemma}[defin]{Lemma}
\newtheorem{theorem}[defin]{Theorem}
\newtheorem{prop}[defin]{Proposition}
\newtheorem{corollary}[defin]{Corollary}

\theoremstyle{definition}

\newtheorem{remark}[defin]{Remark}
\newtheorem{example}[defin]{Example}

\begin{document}
	\begin{titlepage}
		\title{Are You Smarter Than a Random Expert?\\ The Robust Aggregation of Substitutable Signals}
		
		\author{Eric Neyman, Tim Roughgarden}
		
		\date{\today}
		\maketitle
		\thispagestyle{empty}
		
		\begin{abstract}
			The problem of aggregating expert forecasts is ubiquitous in fields as wide-ranging as machine learning, economics, climate science, and national security. Despite this, our theoretical understanding of this question is fairly shallow. This paper initiates the study of forecast aggregation in a context where experts' knowledge is chosen adversarially from a broad class of information structures. While in full generality it is impossible to achieve a nontrivial performance guarantee, we show that doing so is possible under a condition on the experts' information structure that we call \emph{projective substitutes}. The projective substitutes condition is a notion of informational substitutes: that there are diminishing marginal returns to learning the experts' signals. We show that under the projective substitutes condition, taking the average of the experts' forecasts improves substantially upon the strategy of trusting a random expert. We then consider a more permissive setting, in which the aggregator has access to the prior. We show that by averaging the experts' forecasts and then \emph{extremizing} the average by moving it away from the prior by a constant factor, the aggregator's performance guarantee is substantially better than is possible without knowledge of the prior. Our results give a theoretical grounding to past empirical research on extremization and help give guidance on the appropriate amount to extremize.
		\end{abstract}
	\end{titlepage}

	\input{Intro.tex}
	\input{Prelims.tex}
	\input{Prior_Free.tex}
	\input{Known_Prior.tex}
	
	\printbibliography
	
	\appendix
	\input{Appendix.tex}
\end{document}

%% file: Intro.tex
\section{Introduction}
Suppose that you wish to estimate how much the GDP of the United States will grow next year: perhaps you are making financial decisions and want to know whether to expect a downturn. You don't personally know much about the question --- just that the historical average rate of GDP growth has been 3\% --- but on the internet you find several forecasts made by machine learning models. One model predicts 3.5\% growth next year; another predicts 1.5\%; a third predicts a downturn: -1\% growth. How might you take this information into account and turn it into one number: your best guess, all things considered?

Because of the ubiquity of its applications, forecast aggregation is of critical importance to many fields: machine learning, operations research, economics, climate science, epidemiology, and national security, to name a few. Despite this, the theoretical tools we have for understanding this problem are fairly limited.

What should we ask of a framework for comparing competing aggregation methods?
First, for each fixed setup, it should enable the quantitative
assessment of
an aggregation method based on its performance relative to a natural
benchmark (analogous to, for example, assessing an online learning
algorithm via its regret with respect to the best fixed action in
hindsight).
Second,
the framework should be general: rather than evaluating an aggregation method based on its performance under a particular assumption about the experts' information sets, it should assess the method based on its performance over a broad range of possible setups.

We can model each expert as having partial information over the state of the world, and thus the quantity being estimated (which we denote $Y$). The experts' information sets may overlap in essentially arbitrary ways, which we formalize using the notion of an \emph{information structure}.

No aggregation method is simultaneously optimal for every information structure. As such, it is natural to ask which aggregation method optimizes worst-case performance over a broad class of information structures. This is the approach we take, because it has the aforementioned advantages: it assesses aggregation methods based on their performance, but does so broadly rather than under specific assumptions.

\subsection{Our Results}
Without any conditions on the experts' information structure, no aggregation strategy can achieve a nontrivial performance guarantee.\footnote{For example, consider the ``XOR information structure" in which two experts receive independent, random bits, and $Y$ is their XOR. See Section~\ref{sec:improving} for further discussion.} In this work, we optimize for worst-case performance over all information structures that satisfy a condition that we call \emph{projective informational substitutes}. Roughly speaking, experts' signals are informational substitutes if the value of learning an additional signal has diminishing marginal returns. The projective substitutes condition is a particular formalization of this concept that builds on an alternative notion (``weak substitutes") introduced in \cite{cw16} --- a notion that proves inadequate for our purposes.\footnote{In Section~\ref{sec:random_expert} we introduce a ``secret sharing" information structure that shows that with no further assumptions beyond the weak substitutes condition, no aggregation strategy achieves a better performance guarantee than the strategy of choosing a random expert to trust.} Intuitively, substitutable signals allow for effective aggregation because signal interactions are more predictable, so it is possible to infer more from forecasts alone without knowing the information structure.

We consider two settings: the \emph{prior-free setting} and the \emph{known prior setting}. In the prior-free setting, an aggregator receives only the experts' forecasts as input; in the known prior setting, the aggregator additionally knows the prior, i.e. the overall expected value of $Y$ (3\% in our leading example). In both settings, the expert must then output an aggregate forecast.

One simple strategy is to pick an expert at random and ``aggregate" by outputing that expert's forecast. In expectation, this aggregate performs at least as well as the prior; and under the weak substitutes condition of \cite{cw16}, the strategy does at least $1/k$ times as well as someone who knew every expert's signal and the information structure, where $k$ is the number of experts.\footnote{We judge the performance of an aggregation strategy based on its improvement over the prior.} That is, choosing a random expert attains an \emph{approximation ratio} of $1/k$. Unfortunately, we exhibit an information structure that satisfies weak substitutes but on which no aggregation strategy can outperform a $1/k$-approximation (even in the known prior setting).

However, under our slightly stronger assumption of projective substitutes, it is possible to improve upon this $1/k$ baseline. Thus, while one can ask about robust aggregation in many different settings, the projective substitutes condition appears to be a sweet spot: it allows for a broad array of possible information structures while still allowing at nontrivial performance guarantees in both the prior-free and known prior settings. These results are summarized in Figure~\ref{fig:results}.\\

\begin{figure}
	\centering
	\includegraphics[scale=0.8]{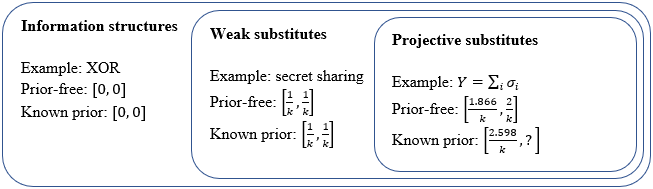}
	\caption{We are interested in the innermost setting, where nontrivial positive results are possible. Intervals given for each setting indicate positive and negative results, respectively (and are stated as asymptotic approximations for the innermost setting).}
	\label{fig:results}
\end{figure}

\noindent In Section~\ref{sec:prior_free}, we investigate the prior-free setting. In this setting, we show that one can improve upon the random expert strategy under the projective substitutes condition, and also that our bound is tight for two experts and close to tight in the general case.
\begin{itemize}[leftmargin=*]
	\item \textbf{(Theorem~\ref{thm:prior_free_positive}) By taking the average of the experts' forecasts, it is possible to attain an approximation ratio of at least $\mathbf{(1 + \sqrt{3}/2)/k - O \parens{1/k^2} \approx 1.866/k}$.} In other words, knowing nothing about an information structure other than the fact that it satisfies the projective substitutes condition, one can significantly improve upon the $1/k$-approximation guarantee of choosing a random expert.
	
	To prove this result, we use the projective substitutes condition to show that one of two things must be true: either (a) the experts' forecasts are (in expectation) fairly different from each other, or (b) the forecasts are somewhat accurate, meaning that they improve substantially upon the prior. In case (a), averaging the experts' forecasts guarantees substantial improvement upon a random forecast; in case (b), even though averaging the forecasts does not improve substantially upon a random forecast, a random forecast already substantially outperforms the prior.
	
	\item \textbf{(Theorem~\ref{thm:prior_free_negative}) There is no aggregation method that attains an approximation ratio of more than $\mathbf{2/k - 1/k^2}$.} In fact, we show that our negative result holds for the following simple class of information structures (which satisfy projective substitutes): each expert receives a numerical signal $\sigma_i$ drawn independently from a normal distribution with some mean $\mu$ and variance $1$, and the quantity being forecast is $Y := \sum_i \sigma_i$. (Thus, this negative result holds for every restriction on information structures that accommodates this class of examples, not only for projective substitutes.)
	
	For any such information structure, expert $i$'s forecast is $Y_i := \sigma_i + (k - 1)\mu$. This means that $Y = \sum_i Y_i - k(k - 1)\mu$. Thus, an aggregator who knew the mean $\mu$ could predict $Y$ perfectly; however, in the prior-free setting, the aggregator does not know $\mu$.
	
	We show (in a formal sense) that for this one-parameter family of information structures, the best strategy is to average the forecasts. Put otherwise, despite the amount of knowledge that the aggregator has (namely, everything other than $\mu$), the aggregator cannot do better than reporting the average of the signals. This strategy attains a $2/k - 1/k^2$ approximation guarantee on these information structures.
	
	\item \textbf{(Theorem~\ref{thm:prior_free_negative_n2}) In the case of $\mathbf{k = 2}$ experts, the positive result given in Theorem~\ref{thm:prior_free_positive} --- which in this case is $\mathbf{\frac{3 + \sqrt{7}}{8} \approx 0.706}$ --- is tight (i.e. no aggregation method can do better).} We exhibit two information structures that, despite having rather different optimal aggregation strategies, the aggregator cannot distinguish between.
\end{itemize}

\noindent In Section~\ref{sec:known_prior}, we investigate the known prior setting. We prove that it is possible to improve upon the aforementioned guarantee of the prior-free setting by \emph{extremizing} the average of the experts' beliefs, i.e. moving it away from the prior. Previous work on forecasting has demonstrated an empirical case for extremization \cite{sbfmtu14} \cite{bmtsu14} \cite{su15}. However, there is comparatively little work grounding extremization from a theoretical standpoint. This work shows that in our robust aggregation setup, the aggregator benefits from extremizing the average. Additionally, our results suggest a particular amount by which to extremize. Specifically, we show that:
\begin{itemize}[leftmargin=*]
	\item \textbf{(Theorem~\ref{thm:known_prior_positive}) By extremizing the average of the experts' forecasts by a particular constant factor that depends on $\mathbf{k}$,\footnote{This factor approaches $\sqrt{3}$ as $k$ approaches infinity.} it is possible to attain an approximation ratio of at least $\mathbf{(3\sqrt{3}/2)/k - O(1/k^2) \approx 2.598/k}$.} The key approach is to introduce a new degree of freedom --- the \emph{extremization factor} --- to optimize over. Setting this factor equal to $1$ returns the optimization problem solved in the proof of Theorem~\ref{thm:prior_free_positive}, but allowing the factor to take on arbitrary values allows us to substantially improve the bound.
	
	\item \textbf{(Theorem~\ref{thm:known_prior_conditional_negative}) It is not possible to attain an approximation ratio of more than $\mathbf{4k/(k + 1)^2}$ using the above approach of averaging and extremizing by a constant factor.} This result holds even in the case when each expert receives a signal from the standard Gaussian distribution, and all that is unknown is whether the signals are fully correlated or fully independent. In the first case, it is optimal not to extremize; in the second, it is optimal to extremize a lot. We prove that no fixed extremization constant attains a better approximation guarantee than stated on both information structures.
	
	\item \textbf{(Theorem~\ref{thm:known_prior_negative_n2}) In the case of $\mathbf{k = 2}$ experts, the positive result given in Theorem~\ref{thm:known_prior_positive} --- which in this case is $\mathbf{\frac{7\sqrt{7} - 17}{2} \approx 0.760}$ --- is tight (i.e. no aggregation method can do better).} As with Theorem~\ref{thm:prior_free_negative_n2}, the idea is to exhibit two information structures that are indistinguishable to the expert but have substantially different optimal aggregation strategies.
\end{itemize}

While we have stated some of our results as asymptotic approximations in $k$ for convenience, these asymptotics are not our focus. Instead, our goal is to understand which \emph{methods} of aggregation work well in which settings. When is the tried and true method of averaging forecasts about optimal, and when is it possible to attain a substantial improvement? Figure~\ref{fig:results_small_k} summarizes our findings for $k = 2 \dots 7$ (which are plausible values for many of the applications that motivate this work). The high-level takeaways are:
\begin{itemize}
	\item Under the projective substitutes condition, it is possible to improve substantially upon selecting a random expert simply by averaging the experts' forecasts.
	\item When only the forecasts are known, no technique can substantially improve upon averaging.
	\item But when the prior is known, extremizing appropriately is substantially better than averaging, and in fact better than any possible aggregation strategy that does not use the prior.
\end{itemize}

\begin{figure}
	\centering
	\begin{tabular}{c|c|cc|c}\
		&Weak subs.&\multicolumn{2}{c}{Proj. subs. (prior-free)} & Proj. subs. (known prior)\\
		k & Pos.\ \& neg.\ & Averaging (positive) & Negative & Extremizing (positive)\\
		\hline
		2&0.500&0.706&0.706&0.760\\
		3&0.333&0.520&0.556&0.596\\
		4&0.250&0.409&0.438&0.488\\
		5&0.200&0.336&0.360&0.412\\
		6&0.167&0.285&0.306&0.356\\
		7&0.143&0.248&0.265&0.314
	\end{tabular}
	\caption{For $k = 2 \dots 7$: the approximation ratio guaranteed by choosing a random expert, which is the best one can do under weak substitutes, followed by some of our positive and negative results under the projective substitutes condition.}
	\label{fig:results_small_k}
\end{figure}

\subsection{Related Work}
\paragraph{Axiomatic and Bayesian approaches} Prior work on forecast aggregation tends to come in two flavors: axiomatic and Bayesian. The axiomatic approach to forecast aggregation seeks to define desirable properties of aggregation methods, and asks which methods satisfy these properties. For example, \cite{aw80} show that linear pooling (i.e. averaging) is the only aggregation method that satisfies both the unanimity preservation and eventwise independence axioms. Other work on the axiomatic approach includes \cite{gen84} and \cite{dl14}.

One drawback of the axiomatic approach is that aggregation methods preferred by these approaches rarely work as characterizations of how one ought to aggregate predictions under specific (natural) information structures. By contrast, the Bayesian approach seeks to create models of experts' information and analyzes the correct way for a Bayesian aggregator to take the information into account. Prior work on aggregation with a Bayesian flavor includes \cite{winkler81}, \cite{mz93}, \cite{sbfmtu14}, \cite{fck15}, and \cite{lgjw17}.

However, most prior work on Bayesian aggregation takes a \emph{parametric} view of aggregation: experts are modeled as Bayesians whose signals are drawn from a particular distribution with one or more parameters, and an aggregation method is chosen to optimize an objective function within the model. For example, \cite{lgjw17} considers a model in which experts have some shared information and additionally receive private samples that are independent and identically distributed according to an exponential family. While this is a natural model, it rests on specific assumptions about the information structure, and results in this model may not generalize well.

\paragraph{Robust aggregation} Our model can be thought of as a hybrid of the axiomatic and Bayesian approaches, blending what we believe to be the most appealing parts of each. We draw from the Bayesian approach in using information structures as a formalism for the experts' knowledge, whereas the goal of producing a single role that satisfies some global property (in our case, worst-case optimality) is reminiscent of the axiomatic approach. Our model is non-parametric: rather than assuming a parameterized family of distributions, we seek to optimize our aggregation method against a broad class of information structures.

Our work is most similar to \cite{abs18}, which likewise seeks to optimize an aggregation method against an adversarially selected information structure. However, the class of information structures that we consider is broader: while they consider the case of two Blackwell-ordered experts (i.e. two experts, an unknown one of whom knows strictly more than the other) and two conditionally independent experts, we consider an arbitrary number of experts from any information structure that satisfies the projective substitutes condition. \cite{lr20} have a similar model, but are also quite restrictive in terms of the information structures they consider. \cite{dil21} uses a similar model, but more distantly related: they consider arbitrary decision problems but restrict the aggregator to a finite number of decisions (just two decisions for many of their results) --- in our setting this would mean forcing the aggregator to choose among finitely many output choices.

Another important difference is that most of the previously mentioned work specifically considers the aggregation of \emph{probabilistic} forecasts, whereas we are interested in aggregating expected value forecasts for arbitrary real-valued quantities.

\paragraph{Extremization} Past empirical work has demonstrated that extremizing the average of the experts' forecasts often improves the aggregate forecast \cite{sbfmtu14} \cite{bmtsu14} \cite{su15}. \cite{bmtsu14} explains this by noting that any individual forecaster should incorporate the fact that they may be missing useful information available to other forecasters, and that simply averaging forecasts would fail to incorporate the full wisdom of the crowd. Studying aggregation in the context of information structures as well, \cite{su15} notes that the forecast average \emph{lacks resolution}, meaning that its variance is provably too low. However, the authors note that the information structure framework is in full generality ``too abstract to be applied in practice" and instead optimize their extremization factor for multivariate Gaussian distributions. On the other hand, our approach of robust aggregation is able not only to provide a theoretical justification for extremization, but also to suggest a particular factor of extremization (Theorem~\ref{thm:known_prior_positive}), thus giving rigorous backing to what had previously been justified either by empirical heuristics or by optimization over a quite narrow class of information structures.

\paragraph{Informational substitutes} In this work we explore notions
of informational substitutes and their relation to forecast
aggregation. The concept of informational substitutes was introduced
in \cite{bhk13} and refined by \cite{cw16}.\footnote{We recommend the
  ArXiv version of \cite{cw16} for the most up-to-date introduction to
  informational substitutes.} We build on \cite{cw16} by introducing
our own notion of substitutes, projective substitutes. Contemporaneous
work \cite{fnw21} explores informational substitutes, though in the
context of experts exchanging information to reach agreement.

%% file: Prelims.tex
\section{Preliminaries} \label{sec:prelims}
\subsection{Information Structures}
We consider a set $\Omega$ of possible states of the world, with a probability distribution $\PP$ over these states. Additionally there are $k$ experts labeled $1 \dots k$. Each expert $i$ learns the value of a random variable $\sigma_i: \Omega \to S_i$; we call $\sigma_i$ expert $i$'s \emph{signal} and $S_i$ their \emph{signal set}. We let $S := S_1 \times \dots \times S_k$ denote the space of signal tuples. An expert's signal can be thought of as partial information about the true state of the world $\omega \in \Omega$. Additionally, we consider a random variable $Y: \Omega \to \RR$ with $\EE{Y^2} < \infty$ whose value we wish to estimate. We use the term \emph{information structure} to refer to the tuple $\mathcal{I} := (\Omega, \PP, S, Y)$.

In the case that $Y$ is determined by the experts' signals (i.e. for every profile of signals there is only one possible value of $Y$), we may summarize $\mathcal{I}$ with two tables, one showing $Y$ as a function of $(\sigma_1, \dots, \sigma_k)$ and the other showing $\pr{(\sigma_1, \dots, \sigma_k)}$. For example, consider the following information structure.

\[
\left\{Y = \begin{tabular}{c|cc}
	&$\sigma_2 = 0$&$\sigma_2 = 1$\\
	\hline
	$\sigma_1 = 0$&0&1\\
	$\sigma_1 = 1$&1&0
\end{tabular}
\qquad \PP =
\begin{tabular}{c|cc}
	&$\sigma_2 = 0$&$\sigma_2 = 1$\\
	\hline
	$\sigma_1 = 0$&1/4&1/4\\
	$\sigma_1 = 1$&1/4&1/4
\end{tabular}\right\}
\]

We refer to this as the \emph{XOR information structure}, because it describes the following situation: two experts receive independent, uniformly random bits, and $Y$ is the XOR of these bits.\footnote{Note that when we write $\sigma_1 = 0$, $0$ is merely a label for a particular signal value in $S_1$; we could have labeled expert 1's signals $a$ and $b$ instead. Our choice of labels reflects the intuition of the XOR structure describing $Y$.}

Consider a subset of experts $A \subseteq [k]$. We define $Y_A := \EE{Y \mid \sigma_i: i \in A}$. That is, $Y_A$ is the random variable whose value is the expectation of $Y$ given the signals of the experts in $A$. If $A = \{i\}$, we write $Y_i$ in place of $Y_{\{i\}}$. With the XOR information structure, $Y_1$ is a random variable whose value is always $\frac{1}{2}$ (since $\EE{Y \mid \sigma_1 = 0} = \EE{Y \mid \sigma_1 = 1} = \frac{1}{2}$). On the other hand, $Y_{\{1, 2\}}$ has value $0$ when $(\sigma_1, \sigma_2)$ is either $(0, 0)$ or $(1, 1)$, and $1$ otherwise. The random variable $Y_\emptyset$ is simply $\EE{Y}$, i.e. the unconditional expected value of $Y$. $Y_\emptyset$ can be thought of as the prior on $Y$.

\subsection{Improving on the Prior} \label{sec:improving}
In this work, we will be taking the perspective of an aggregator who receives estimates of $Y$ from each expert. The aggregator then produces an estimate $X$ of $Y$ which is as accurate as possible. In particular, we care about the \emph{robust} estimation of $Y$: a single estimate that is simultaneously as accurate as possible across all possible information structures (satisfying the projective substitutes condition, which we discuss below).

We assess an aggregator's performance by the squared distance between their estimate $X$ and the true value $Y$. That is, the aggregator wishes to minimize $\EE{(Y - X)^2}$. We define the function $v(X)$ as follows to reflect the \emph{quality of $X$ as an estimate of $Y$}.

\begin{defin}
	Given an information structure $\mathcal{I} = (\Omega, \PP, S, Y)$ and a random variable $X$, we define
	\[v(X) := \EE{(Y - \EE{Y})^2} - \EE{(Y - X)^2}.\]
\end{defin}

Thus, $v(\cdot)$ is the improvement in loss provided by $X$ over an uninformed estimate. For example, $v(Y_\emptyset) = 0$ and $v(Y) = \EE{(Y - \EE{Y})^2}$ is the variance of $Y$. We cannot possibly hope for any $X$ such that $v(X) > Y_{[k]}$, since $Y_{[k]}$ is the estimate produced by knowing all information that exists. This motivates comparing $v(X)$ against the benchmark $v(Y_{[k]})$.

However, the aggregator does not know the underlying information structure --- only the experts' estimates. Specifically, we will consider two settings:
\begin{enumerate}[label=(\arabic*)]
	\item The \emph{prior-free setting}: the aggregator's estimate is only based on the experts' estimates. That is, $X$ is a function of $Y_1, \dots, Y_k$.
	\item The \emph{known prior setting}: the aggregator knows the experts' estimates and the prior. That is, $X$ is a function of $Y_1, \dots, Y_k$ and $\EE{Y}$.
\end{enumerate}
That is, $X$ is a function of $k$ real numbers (or $k + 1$, in the known prior setting); we call this function the aggregator's \emph{aggregation strategy}. The aggregator's goal is to come up with an aggregation strategy that performs well across information structures (we formalize this below).\\

In the known prior setting, the aggregator can report $X = \EE{Y}$; then $v(X) = 0$ (we call this the \emph{trivial aggregation strategy}). In both settings it is possible to do at least as well by reporting e.g. $X = Y_1$. On the other hand, without any conditions on the information structure, it is not always possible to do better: in the XOR information structure (above), the aggregator is guaranteed to receive $Y_1 = Y_2 = \frac{1}{2}$, and it is impossible for the aggregator to improve upon simply reporting the prior of $\frac{1}{2}$.

\subsection{Informational Complements and Substitutes}
Intuitively, in the XOR information structure, the aggregator is impeded by the fact that the experts' signals are \emph{informational complements}: each signal (and the estimate it produces) is not valuable by itself, but the two signals are valuable when taken together.

The opposite of informational complements is \emph{informational substitutes}. \cite{cw16} discuss a few different notions of substitutes, of which the most relevant one for us is \emph{weak informational substitutes}. This notion is a property of $v(\cdot)$ as a set function on the subsets of $[k]$. The function is guaranteed to be monotone --- that is, $v(Y_A) \le v(Y_B)$ for $A \subseteq B$ --- whereas the weak substitutes condition additionally requires submodularity.\footnote{The authors define informational substitutes with respect to an arbitrary value function; in our case, the value function is the function $v$.}

\begin{defin}[\cite{cw16}] \label{def:weak_subs}
	We say that $\mathcal{I} = (\Omega, \PP, S, Y)$ for $k$ experts satisfies \emph{weak informational substitutes} if for all $A \subseteq B \subseteq [k]$ and $i \in [k]$, we have
	\begin{equation} \label{eq:weak_subs}
		v(Y_{A \cup \{i\}}) - v(Y_A) \ge v(Y_{B \cup \{i\}}) - v(Y_B).
	\end{equation}
\end{defin}

That is, $\mathcal{I}$ satisfies weak substitutes if the marginal improvement in the ability to estimate $Y$ from learning a signal $i$ is a decreasing function of the subset of signals already known. The XOR information structure does not satisfy weak substitutes: we have $v(\emptyset) = v(Y_1) = v(Y_2) = 0$ and $v(Y_{[2]}) = \frac{1}{4}$, so the marginal value of learning $\sigma_2$ is \emph{higher} if $\sigma_1$ is already known.

By contrast, consider the information structure in which Alice and Bob are given the \emph{same} input bit $b$, and $Y = b$. In this case, the marginal value of Bob's signal is zero if Alice's signal is already known, so this information structure satisfies weak substitutes.

\subsection{Random Expert Strategy Under Weak Substitutes} \label{sec:random_expert}
It is not surprising that with no knowledge of the information structure, it is impossible to outperform the trivial strategy. Perhaps it would be possible to do better with only a coarse constraint on the information structure. It is not \emph{a priori} obvious that this should be possible. However, if $\mathcal{I}$ satisfies weak substitutes, then it is possible to outperform the trivial strategy by reporting a random expert's belief:

\begin{prop} \label{prop:random_expert}
	Suppose that $\mathcal{I} = (\Omega, \PP, S, Y)$ satisfies weak substitutes, and let $X$ be equal to $Y_i$ for a uniformly random $i \in [k]$ (we call this the \emph{random expert strategy}). Then $v(X) \ge \frac{1}{k} v(Y_{[k]})$.
\end{prop}

\begin{proof}
	For $j \in [k]$, plug $A = \emptyset, B = [j - 1], i = j$ into Equation~\ref{eq:weak_subs}. Adding these $k$ inequalities (and recalling that $v(\emptyset) = 0$), we find that $\sum_{j = 1}^k v(Y_j) \ge v(Y_{[k]})$. Therefore, for $X$ as in the proposition statement, we have
	\[v(X) = \frac{1}{k} \sum_{i = 1}^k v(Y_i)  \ge \frac{1}{k} v(Y_{[k]}),\]
	as desired.
\end{proof}

Put otherwise, the random expert strategy attains an \emph{approximation ratio} of $1/k$.

\begin{defin}
	Given an information structure $\mathcal{I} = (\Omega, \PP, S, Y)$ with $k$ experts, the \emph{approximation ratio} of a random variable $Z$ is given by the quantity $v(Z)/v \parens{Y_{[k]}}$.
\end{defin}

Unfortunately, the following result (whose proof Appendix~\ref{appx:prelims_omitted}) shows that with no further assumptions, it is not possible to attain an approximation ratio larger than $1/k$:

\begin{prop} \label{prop:weak_bad}
	For every $k$, there is an information structure that satisfies the weak substitutes condition, such that in both the prior-free and known prior settings, no aggregation strategy attains an approximation ratio greater than $1/k$ on the information structure.
\end{prop}

The key idea is to use Shamir secret sharing \cite{shamir79} to create an $(k, r)$-threshold scheme for a uniformly random $r \in [k]$. Then $v(\cdot)$ is additive (and thus submodular) on the subsets of $[k]$, but an aggregator who only knows the experts' reports will only be able to recover the secret if $r = 1$.

\begin{proof}
	Let $p > k$ be a prime. Consider the following information structure (the \emph{secret sharing information structure}).
	\begin{itemize}
		\item An integer $r \in [k]$ is selected uniformly at random and announced.
		\item A random $(r - 1)$-th degree polynomial $P(x) = a_0 + a_1x + \dots + a_{r - 1}x^{r - 1}$ over $\FF_p$ is selected, with coefficients chosen uniformly at random from $\FF_p$, except that $a_0$ is either $-1$ or $1$ (also uniformly). For each $i \in [k]$, expert $i$ is told $P(i)$.
		\item The quantity $Y$ is equal to $1$ if $a_0 = 1$ and $-1$ if $a_0 = -1$.
	\end{itemize}
	Note that for a fixed choice of $r$, we have $Y_A = 0$ if $\abs{A} < r$ and $Y_A = \pm 1$ if $\abs{A} \ge r$. Therefore, for any $A$ we have that $Y_A = \pm 1$ with probability $\frac{\abs{A}}{k}$ and $0$ otherwise. Therefore, we have that $v(Y_A) = \frac{\abs{A}}{k}$, so $v(\cdot)$ is additive (and thus submodular). Thus, this information structure satisfies weak substitutes.
	
	On the other hand, note that with probability $\frac{k - 1}{k}$, all experts report $0$ to the aggregator, in which case the aggregator cannot do better in expectation than also reporting $0$. Thus, it is impossible for the aggregator to report an estimate $X$ with $v(X) > \frac{1}{k}$, and so an approximation ratio larger than $\frac{1}{k}$ is not attainable.
\end{proof}

The secret sharing information structure is a lottery over $k$ different information structures, for each of which $v(\cdot)$ is a threshold function: any $r - 1$ experts know nothing ($v(A) = 0$ if $\abs{A} < r$), while any $r$ experts know everything ($v(A) = 1$ if $\abs{A} \ge r$). Except for $r = 1$, these information structures have experts that should be intuitively regarded as \emph{complementary}. Indeed, these structures generalize the XOR information structure (which is the case of $k = r = 2$). This suggests that properties of $v(\cdot)$ as a set function on $[k]$ are insufficient to capture what we intuitively mean by substitutable signals. This motivates us to seek a natural but stronger notion of informational substitutes --- one that is well-motivated and not too restrictive, but which rules out information structures such as this one and allows an aggregator to outperform the guarantee of the random expert strategy.

\subsection{Projective Substitutes} \label{sec:proj_subs}
The following fact (which we prove in Appendix~\ref{appx:prelims_omitted}) helps to motivate the notion that we will introduce.

\begin{prop} \label{prop:weak_subs_rewrite}
	The weak substitutes condition may be rewritten as: for any $i$ and $A \subseteq B$, we have
	\[\EE{(Y_B - Y_A)^2} \ge \EE{(Y_{B \cup \{i\}} - Y_{A \cup \{i\}})^2}.\]
\end{prop}

Intuitively, this interpretation of substitutes says: \textbf{For any expert $i$, a set $A$ of experts becomes better at predicting the belief of a superset of experts $B$ if $i$'s signal is announced.} Here, the \emph{belief} of a set $T$ of experts refers to the expected value of $Y$ conditioned on the experts' signals, i.e. $Y_T$.

This matches the intuition of substitutes as diminishing marginal returns: if signal $i$ becomes known, the ``information gap" between $A$ and $B$ decreases.

A more general notion of substitutes would require this to hold \emph{even when $B$ is not a supserset of $A$}. That is: for all $i, A, B$, the $A \cup \{i\}$ can predict the belief of $B \cup \{i\}$ better than $A$ can predict the belief of $B$. This captures the spirit of diminishing marginal returns in a somewhat broader context.

Let us formalize the notion of $A$'s prediction of $B$'s belief. By this we mean the expected value of $Y_B$ given the signal outcomes of the experts in $A$, i.e. $\EE{Y_B \mid \{\sigma_i: i \in A\}}$.

\begin{defin}
	Given an information structure $\mathcal{I} = (\Omega, \PP, S, Y)$ for $k$ experts and subsets $A, B \subseteq [k]$, \emph{$A$'s prediction of $B$'s belief} is defined as the expected value of $Y_B$ given the signal outcomes of the experts in $A$, i.e.
	\[Y_{B \to A} := \EE{Y_B \mid \{\sigma_i: i \in A\}}.\]
\end{defin}

We now state our substitutes definition, which strengthens the weak substitutes condition.
\begin{defin} \label{def:projective_subs}
	An information structure $\mathcal{I} = (\Omega, \PP, S, Y)$ for $k$ experts satisfies \emph{projective substitutes} if for all $A, B \subseteq [k]$ and $i \in [k]$, we have
	\begin{equation} \label{eq:proj_subs}
		\EE{(Y_B - Y_{B \to A})^2} \ge \EE{(Y_{B \cup \{i\}} - Y_{B \cup \{i\} \to A \cup \{i\}})^2}.
	\end{equation}
\end{defin}

The secret sharing information structure does not satisfy projective substitutes: take $A, B, i$ with $\abs{B} \ge \abs{A}$ and $i \in A \setminus B$. On the other hand, the example below does satisfy projective substitutes.

\begin{example}
	Consider the following information structure.
	\[\left\{Y = \begin{tabular}{c|cc}
		&$\sigma_2 = 1$&$\sigma_2 = 2$\\
		\hline
		$\sigma_1 = 1$&0&1\\
		$\sigma_1 = 2$&1&2
	\end{tabular}
	\qquad \PP =
	\begin{tabular}{c|cc}
		&$\sigma_2 = 1$&$\sigma_2 = 2$\\
		\hline
		$\sigma_1 = 1$&0.3&0.2\\
		$\sigma_1 = 2$&0.2&0.3
	\end{tabular}\right\}
	\]
	Let $A = \{1\}$ and $B = \{2\}$. It is not difficult to compute that
	\[Y_B = \begin{tabular}{c|cc}
		&$\sigma_2 = 1$&$\sigma_2 = 2$\\
		\hline
		$\sigma_1 = 1$&0.4&1.6\\
		$\sigma_1 = 2$&0.4&1.6
	\end{tabular} \qquad \qquad
	Y_{B \to A} = \begin{tabular}{c|cc}
		&$\sigma_2 = 1$&$\sigma_2 = 2$\\
		\hline
		$\sigma_1 = 1$&0.88&0.88\\
		$\sigma_1 = 2$&1.12&1.12
	\end{tabular}\]
	For example, in the $(\sigma_1, \sigma_2) = (1, 1)$ case we have that $Y_B = \frac{0.3 \cdot 0 + 0.2 \cdot 1}{0.3 + 0.2} = 0.4$. $Y_B$ is then used to compute $Y_{B \to A}$: for example, in the $(\sigma_1, \sigma_2) = (1, 1)$ case we have that $Y_{B \to A} = \frac{0.3 \cdot 0.4 + 0.2 \cdot 1.6}{0.3 + 0.2} = 0.88$, as this is the expected value of $Y_B$ conditioned on $\sigma_1 = 1$.
	
	It can be computed that $\EE{(Y_B - Y_{B \to A})^2} = 0.3456$, whereas $\EE{(Y_{B \cup \{1\}} - Y_{B \cup \{1\} \to A \cup \{1\}})^2} = \EE{(Y - Y_A)^2} = 0.24$, so Equation~\ref{eq:proj_subs} is satisfied for $A = \{1\}, B = \{2\}, i = 1$. It can be verified that Equation~\ref{eq:proj_subs} is in fact satisfied for \emph{all} $A, B, i$, so this information structure satisfies projective substitutes.
\end{example}

The projective substitutes definition can be interpreted as describing the class of information structures in which full information revelation is a dominant strategy. Specifically, consider a central party (call them the \emph{elicitor}) who knows the information structure but does not know the experts' signals. Experts are truthful, but may be strategic: they will not lie about their signal, but may decide not to reveal it. The elicitor wishes to structure incentives that will encourage each expert to reveal their signal. The elicitor puts experts on teams (but does not immediately announce the teams). Then:
\begin{enumerate}
	\item Each expert either reveals their signal to the elicitor, or does not.
	\item The elicitor announces which experts revealed their signals and announces the teams.
	\item Each team makes a prediction about the elicitor's posterior belief (after learning the signals of all experts who decided to reveal) and is scored using a quadratic scoring rule (i.e. penalized by the squared distance between their prediction of the elicitor's belief and the elicitor's actual belief).\footnote{Why not elicit $Y$ directly? Eliciting each team's best guess about the \emph{elicitor's} belief is particularly compelling in situations in which the true value of $Y$ will never be known, or will be learned in the far future. Under these circumstances, the elicitor's belief serves as an approximation for $Y$ given the available information.}
\end{enumerate}

This mechanism incentivizes experts to reveal their signals if and only if the information structure satisfies projective substitutes. Formally:
\begin{prop}
	An information structure satisfies projective substitutes if and only if in the above mechanism, revealing one's signal is a dominant strategy for every expert, regardless of who is on their team.
\end{prop}

\begin{proof}
	First suppose that the information structure satisfies projective substitutes. Consider any expert $i$, let $A$ be $i$'s team, and let $B$ be the set of all other experts who reveal their signals. If $i$ does not reveal their signal, then the elicitor's belief will be $Y_B$ and $A$'s prediction of the elicitor's belief will be $Y_{B \to A}$. If $i$ reveals their signal, then the elicitor's belief will be $Y_{B \cup \{i\}}$ and $A$'s prediction of the elicitor's belief will be $Y_{B \cup \{i\} \to A} = Y_{B \cup \{i\} \to A \cup \{i\}}$. Therefore, by the projective substitutes, condition, $A$'s expected prediction error is smaller if $i$ reveals their signal to the expert.
	
	Conversely, suppose that the information structure does not satisfy projective substitutes. Then there are sets $A, B \subseteq [k]$ and $i \in A$ (see Remark~\ref{rem:proj_subs_facts}~\ref{item:proj_subs_equiv}) such that
	\[\EE{(Y_B - Y_{B \to A})^2} < \EE{(Y_{B \cup \{i\}} - Y_{B \cup \{i\} \to A \cup \{i\}})^2}.\]
	Consider expert $i$, suppose their team is $A$, and suppose that the set of experts excluding $i$ who reveal their signal is $B$. Then $i$ is incentivized not to reveal their signal to the elicitor, as revealing their signal will increase $A$'s expected prediction error.
\end{proof}

\begin{remark}[Facts about projective substitutes] \label{rem:proj_subs_facts} \phantom{}
	\begin{enumerate}[label=(\roman*)]
		\item The weak substitutes condition is equivalent to Equation~\ref{eq:proj_subs} holding for all $A \subseteq B$, so the projective substitutes condition is stronger. In fact it is strictly stronger, as it excludes the secret sharing information structure.
		
		\item \label{item:proj_subs_equiv} There are several equivalent formulations of projective substitutes. One definition replaces $\{i\}$ with an arbitrary set $X$. Another modifies the definition by only requiring Equation~\ref{eq:proj_subs} to hold if $i \in A$.\footnote{To see that this is equivalent, for any $A, i$ with $i \not \in A$ define $A' := A \cup \{i\}$. Then Equation~\ref{eq:proj_subs} for $A', B, i$ has the same right-hand side but a smaller or equal left-hand side, and is thus more difficult to satisfy.}
		
		\item The notation $Y_{B \to A}$ comes from the fact that, by Theorem~\ref{thm:projection}, $Y_{B \to A}$ is the projection of $Y_B$ onto the sigma algebra generated by the signals of the experts in $A$. As a consequence of this alternative formulation, $Y_{B \to A}$ is the closest random variable (by expected squared distance) to $Y_B$ among all random variables that depend only on the values $(\sigma_i)_{i \in A}$.
	\end{enumerate}
\end{remark}

In Sections~\ref{sec:prior_free} and \ref{sec:known_prior}, we explore how the projective substitutes condition allows us to improve upon the $1/k$-approximation guarantee attainable under the weak substitutes condition.

\subsection{The Pythagorean Theorem}
Finally, we will find the following fact to be useful throughout our work.
\begin{prop}[Pythagorean theorem] \label{prop:pythag}
	Let $A$ be a random variable, $B = \EE{A \mid \mathcal{F}}$ where $\mathcal{F}$ is a sigma-algebra, and $C$ be a random variable defined on $\mathcal{F}$. Then
	\[\EE{(A - C)^2} = \EE{(A - B)^2} + \EE{(B - C)^2}.\]
\end{prop}

Informally, $B$ is the expected value of $A$ conditional on some partial information, and $C$ is a random variable whose value only depends on this information. We defer the proof to Appendix~\ref{appx:prelims_omitted}.

We call Proposition~\ref{prop:pythag} the ``Pythagorean theorem" because it is precisely the Pythagorean theorem in the Hilbert space described by the following seminal theorem from probability theory.

\begin{theorem}[\cite{zidak57}] \label{thm:projection}
	For a given probability space $(\Omega, \mathcal{F}, \PP)$, consider the Hilbert space $\mathcal{L}^2$ of square-integrable random variables with the inner product $\angles{X, Y} := \EE{XY}$. For a given sub-sigma-algebra $\mathcal{F}'$ of the probability space, the map $X \mapsto \EE{X \mid \mathcal{F}'}$ is same as the orthogonal projection map of $\mathcal{L}^2$ onto the subspace of $\mathcal{F}'$-measurable square-integrable random variables.
\end{theorem}

The Pythagorean theorem lets us rewrite the approximation ratio in a more convenient form. (See Appendix~\ref{appx:prelims_omitted} for the proof.)
\begin{corollary} \label{cor:improvement_rewrite}
	The approximation ratio of a random variable $Z$ that depends only on $\sigma_1, \dots, \sigma_k$ is equal to
	\[1 - \frac{\EE{(Y_{[k]} - Z)^2}}{\EE{(Y_{[k]} - \EE{Y})^2}}.\]
\end{corollary}

%% file: Prior_Free.tex
\section{The Prior-Free Setting} \label{sec:prior_free}
We begin our investigation in the prior-free setting and ask: what is the largest approximation ratio that can be guaranteed under the projective substitutes condition? In this section we give a positive result and a negative result. The positive result is that \textbf{averaging the experts' reports attains an approximation ratio of at least $\mathbf{(1 + \sqrt{3}/2)/k - O(1/k^2) \approx 1.866/k}$}. The negative result is that for all $k$, \textbf{no aggregation method attains an approximation ratio of more than $\mathbf{2/k - 1/k^2}$}. Thus, projective substitutes enables a significant improvement over the $1/k$-approximation guarantee of the random expert strategy, but no more than by a factor of two.

\subsection{Positive Result for the Prior-Free Setting}
\begin{theorem} \label{thm:prior_free_positive}
	Let $\mathcal{I} = (\Omega, \PP, S, Y)$ be an information structure for $k$ experts that satisfies projective substitutes, and let $X = \frac{1}{k} \sum_{i = 1}^k Y_i$. Then $X$ attains an approximation ratio of at least
	\[\frac{2}{k} - \frac{k - 1}{2k(2k - 1 + \sqrt{3k^2 - 3k + 1})} - \frac{1}{k^2} \ge \parens{1 + \frac{\sqrt{3}}{2}} \cdot \frac{1}{k} - O \parens{\frac{1}{k^2}}.\]
\end{theorem}

\begin{proof}
	Let $X = \frac{1}{k} \sum_{i = 1}^k Y_i$. By Corollary~\ref{cor:improvement_rewrite}, showing that $X$ achieves an approximation ratio of $\alpha$ is equivalent to showing that
	\begin{equation} \label{eq:alpha_rewrite}
		\EE{(Y_{[k]} - X)^2} \le (1 - \alpha)\EE{(Y_{[k]} - \EE{Y})^2}.
	\end{equation}
	The first step in our proof uses the following fact: for any numbers $y_1, \dots, y_k$ and $y$, we have
	\[\parens{y - \frac{y_1 + \dots + y_k}{k}}^2 = \frac{1}{k} \sum_{i = 1}^k (y - y_i)^2 - \frac{1}{k^2} \sum_{1 \le i < j \le k} (y_i - y_j)^2.\]
	This equality follows from rearranging terms, and applying it in expectation for $y = Y_{[k]}$ and $y_i = Y_i$ gives us the following equality.
	\begin{equation} \label{eq:formula_one}
		\EE{(Y_{[k]} - X)^2} = \frac{1}{k} \sum_{i = 1}^k \EE{(Y_{[k]} - Y_i)^2} - \frac{1}{k^2} \sum_{1 \le i < j \le k} \EE{(Y_i - Y_j)^2}.
	\end{equation}
	The left-hand side here is the same as in Equation~\ref{eq:alpha_rewrite}; meanwhile, the right-hand side has a term representing the average error of a random expert and another term representing the average expected distance between the experts' forecasts. The following lemma allows us to get a handle on this last term.
	
	\begin{lemma} \label{lem:ab}
		For all $i, j$, and for all $a, b \ge 0$ such that $b \ge \frac{(2a - 1)^2}{4a}$, we have
		\begin{equation} \label{eq:ab}
			\EE{(Y_i - Y_j)^2} \ge a \parens{\EE{(Y_{\{i,j\}} - Y_i)^2} + \EE{(Y_{\{i,j\}} - Y_j)^2}} - b \parens{\EE{(Y_i - \EE{Y})^2} + \EE{(Y_j - \EE{Y})^2}}.
		\end{equation}
	\end{lemma}

	The proof of Lemma~\ref{lem:ab} relies on the projective substitutes assumption. The lemma lets us flexibly lower bound the expected distance between $Y_i$ and $Y_j$ in terms of the average expected distance from $Y_{\{i, j\}}$ to $Y_i$ and $Y_j$. Intuitively, the projective substitutes condition guarantees such a bound because expert $i$ must be able to forecast $Y_{\{i, j\}}$ better than they can forecast $Y_j$. For now we assume the truth of Lemma~\ref{lem:ab} and return to the proof of Theorem~\ref{thm:prior_free_positive}. We note that we may rewrite
	
	\begin{equation} \label{eq:pythag_app}
		\EE{(Y_i - \EE{Y})^2} = \EE{(Y_{[k]} - \EE{Y})^2} - \EE{(Y_{[k]} - Y_i)^2},
	\end{equation}
	by the Pythagorean theorem. Additionally, we note that by weak substitutes (which follows from projective substitutes), for all $i$ we have
	\begin{equation} \label{eq:weak_subs_use}
		\sum_{j \neq i} \EE{(Y_{\{i, j\}} - Y_i)^2} \ge \EE{(Y_{[k]} - Y_i)^2}.
	\end{equation}
	To see this, consider for example $i = 1$. By weak substitutes, $\EE{(Y_{\{1, j\}} - Y_1)^2} \ge \EE{(Y_{[j]} - Y_{[j - 1]})^2}$, so the left-hand side of Equation~\ref{eq:weak_subs_use} is greater than or equal to $\sum_{j > 1} \EE{(Y_{[j]} - Y_{[j - 1]})^2}$, which (by $k - 2$ applications of the Pythagorean theorem) is equal to $\EE{(Y_{[k]} - Y_1)^2}$. Now, combining Equations~\ref{eq:formula_one}, \ref{eq:ab}, \ref{eq:pythag_app}, and \ref{eq:weak_subs_use} gives us that
	\begin{equation} \label{eq:almost_done}
		\EE{(Y_{[k]} - X)^2} \le \frac{b(k - 1)}{k} \EE{(Y_{[k]} - \EE{Y})^2} + \parens{\frac{1}{k} - \frac{a}{k^2} - \frac{b(k - 1)}{k^2}} \sum_{i = 1}^k \EE{(Y_{[k]} - Y_i)^2}.
	\end{equation}
	for any $a, b$ satisfying Lemma~\ref{lem:ab}. Now, note that by weak substitutes we have
	\begin{equation} \label{eq:conditional_step}
		\sum_{i = 1}^k \EE{(Y_{[k]} - Y_i)^2} = k \EE{(Y_{[k]} - \EE{Y})^2} - \sum_{i = 1}^k \EE{(Y_i - \EE{Y})^2} \le (k - 1) \EE{(Y_{[k]} - \EE{Y})^2},
	\end{equation}
	where the first step uses the Pythagorean theorem and the second step follows from Proposition~\ref{prop:random_expert}. Therefore, if $\frac{1}{k} - \frac{a}{k^2} - \frac{b(k - 1)}{k^2} \ge 0$, we may write Equation~\ref{eq:almost_done} as
	\[\EE{(Y_{[k]} - X)^2} \le \frac{k - 1}{k} \parens{1 - \frac{a - b}{k}} \EE{(Y_{[k]} - \EE{Y})^2}.\]
	To make this inequality as tight as possible, we wish to make $a - b$ as large as possible; our constraints are that $b \ge \frac{(2a - 1)^2}{4a}$ and $\frac{1}{k} - \frac{a}{k^2} - \frac{b(k - 1)}{k^2} \ge 0$. The optimal values are
	\[a = \frac{2k - 1 + \sqrt{3k^2 - 3k + 1}}{2k} \text{ and } b = \frac{(2a - 1)^2}{4a}.\]
	This gives us
	\[\alpha = 1 - \frac{k - 1)}{k} \parens{1 - \frac{a - b}{k}} = \frac{2}{k} - \frac{k - 1}{2k(2k - 1 + \sqrt{3k^2 - 3k + 1})} - \frac{1}{k^2},\]
	as desired.
\end{proof}

\begin{proof}[Proof of Lemma~\ref{lem:ab}]
	Recall from Theorem~\ref{thm:projection} that random variables can be thought of as points in a Hilbert space. Let $T_i$ be the projection of $Y_j$ onto the space of all affine combinations of $Y_i$ and $Y_{\emptyset}$, i.e. $\{\beta Y_i + (1 - \beta) Y_{\emptyset}: \beta \in \RR\}$. (Recall that $Y_{\emptyset}$ is the random variable that is always equal to $\EE{Y}$.) Define $T_i$ analogously. Note that $\EE{(Y_j - Y_{j \to i})^2} \le \EE{(Y_j - T_j)^2}$, since $Y_{j \to i}$ is the closest point to $Y_j$ of the subspace of random variables that depend only on $\sigma_i$, and the aforementioned affine space is a subset of that subspace. Additionally, by projective substitutes (with $A = \{i\}, B = \{j\}, i = i$ in Definition~\ref{def:projective_subs}), we have that $\EE{(Y_{\{i, j\}} - Y_i)^2} \le \EE{(Y_j - Y_{j \to i})^2}$. Therefore, we have that $\EE{(Y_{\{i, j\}} - Y_i)^2} \le \EE{(Y_j - T_j)^2}$, and similarly that $\EE{(Y_{\{i, j\}} - Y_j)^2} \le \EE{(Y_i - T_i)^2}$. It therefore suffices to show that
	\begin{equation} \label{eq:wish_to_show_1}
		\EE{(Y_i - Y_j)^2} \ge a \parens{\EE{(Y_i - T_i)^2} + \EE{(Y_j - T_j)^2}} - b \parens{\EE{(Y_i - Y_{\emptyset})^2} + \EE{(Y_j - Y_{\emptyset})^2}}.
	\end{equation}
	
	By the Pythagorean theorem,\footnote{While in most cases by ``Pythagorean theorem" we mean Proposition~\ref{prop:pythag}, in this case we are referring to the fact that for orthogonal vectors $\vect{x}$ and $\vect{y}$, we have $\norm{\vect{x}}^2 + \norm{\vect{y}}^2 = \norm{\vect{x} + \vect{y}}^2$ (and applying this fact to e.g. $\vect{x} = Y_i - T_i$, $\vect{y} = Y_j - T_i$). Proposition~\ref{prop:pythag} can be thought of as following from this fact (in the context of Theorem~\ref{thm:projection}).} we know the following four facts.
	\begin{align*}
		\EE{(Y_i - Y_j)^2} &= \EE{(Y_i - T_i)^2} + \EE{(Y_j - T_i)^2}; & \EE{(Y_i - Y_j)^2} &= \EE{(Y_j - T_j)^2} + \EE{(Y_i - T_j)^2}\\
		\EE{(Y_i - Y_{\emptyset})^2} &= \EE{(Y_i - T_i)^2} + \EE{(T_i - Y_{\emptyset})^2}; & \EE{(Y_j - Y_{\emptyset})^2} &= \EE{(Y_j - T_j)^2} + \EE{(T_j - Y_{\emptyset})^2}
	\end{align*}
	These let us rewrite Equation~\ref{eq:wish_to_show_1} as such:
	\begin{align} \label{eq:wish_to_show_2}
		&\frac{1}{2} \parens{\EE{(Y_j - T_i)^2} + \EE{(Y_i - T_j)^2}} + \parens{a - \frac{1}{2}} \parens{\EE{(T_i - Y_{\emptyset})^2} + \EE{(T_j - Y_{\emptyset})^2}} \nonumber\\
		&\ge \parens{a - b - \frac{1}{2}} \parens{\EE{(Y_i - Y_{\emptyset})^2} + \EE{(Y_j - Y_{\emptyset})^2}}.
	\end{align}
	We wish to show that this inequality holds so long as $b \ge \frac{(2a - 1)^2}{4a}$. To do so, we note the following fact: for any random variables $Q, R$ and non-negative reals $c_1, c_2$, we have that
	\[c_1 \EE{Q^2} + c_2 \EE{R^2} \ge \frac{c_1 c_2}{c_1 + c_2} \EE{(Q + R)^2}.\]
	This follows (after multiplying through by $c_1 + c_2$ and cancelling terms) from the fact that for all $q, r$ we have $(c_1 q)^2 + (c_2 r)^2 \ge 2 (c_1 q)(c_2 r)$. Now, we apply this identity to $Q := Y_j - T_i$ and $R := T_i - Y_{\emptyset}$, with $c_1 = \frac{1}{2}$ and $c_2 = a - \frac{1}{2}$, and also to $Q := Y_i - T_j$ and $R := T_j - Y_{\emptyset}$. This tells us that
	\begin{align*}
		&\frac{1}{2} \parens{\EE{(Y_j - T_i)^2} + \EE{(Y_i - T_j)^2}} + \parens{a - \frac{1}{2}} \parens{\EE{(T_i - Y_{\emptyset})^2} + \EE{(T_j - Y_{\emptyset})^2}}\\
		&\ge \frac{2a - 1}{4a} \parens{\EE{(Y_i - Y_{\emptyset})^2} + \EE{(Y_j - Y_{\emptyset})^2}}.
	\end{align*}
	Therefore, Equation~\ref{eq:wish_to_show_2} holds so long as $\frac{2a - 1}{4a} \ge a - b - \frac{1}{2}$, which is equivalent to $b \ge \frac{(2a - 1)^2}{4a}$.
\end{proof}

\subsection{Negative Results for the Prior-Free Setting}
\begin{theorem} \label{thm:prior_free_negative}
	Fix any $k \ge 1$. For every $\mu \in \RR$, let $\mathcal{I}_\mu$ be the information structure defined as follows: each expert $i$ receives an independent signal $\sigma_i \sim N(\mu, 1)$, and $Y := \sum_{i = 1}^k \sigma_i$. Then
	\begin{enumerate}[label=(\roman*)]
		\item $\mathcal{I}_\mu$ satisfies projective substitutes for all $\mu$. \label{item:negative_1}
		\item No aggregation strategy achieves an approximation ratio of more than $\frac{2}{k} - \frac{1}{k^2}$ on every $\mathcal{I}_\mu$. \label{item:negative_2}
	\end{enumerate}
\end{theorem}

We note that there may be stronger substitutes conditions (or other conditions) that the information structures $\mathcal{I}_{\mu}$ satisfy. The negative result of theorem~\ref{thm:prior_free_negative} carries over to any such setting.

\begin{proof}
	We first prove \ref{item:negative_1}. Without loss of generality, assume that $\mu = 0$. For any $A, B$, we have $Y_B = \sum_{j \in B} \sigma_j$ and $Y_{B \to A} = \sum_{j \in A \cap B} \sigma_j$ (from the perspective of the experts in $A$, the values $\sigma_j$ for $j \in A \cap B$ are known and the values $\sigma_j$ for $j \in B \setminus A$ are zero in expectation). Therefore, for any $A, B, i$ we have
	\[\EE{(Y_B - Y_{B \to A})^2} = \parens{\sum_{j \in B \setminus A} \sigma_j}^2 = \parens{\sum_{j \in (B \cup \{i\}) \setminus (A \cup \{i\})} \sigma_j}^2 = \EE{(Y_{B \cup \{i\}} - Y_{B \cup \{i\} \to A \cup \{i\}})^2},\]
	so Equation~\ref{eq:proj_subs} is satisfied (with equality).\\
	
	We now prove \ref{item:negative_2}, and we do so in two steps: first we show that taking the average of the experts' reports yields an approximation ratio of exactly $\frac{2}{k} - \frac{1}{k^2}$ for all $\mu$, and second we show that no aggregation method beats averaging for every $\mu$.
	
	For the first step, again assume without loss of generality that $\mu = 0$. Then $Y_{[k]} = Y = \sum_i \sigma_i$ and the average of the $Y_i$'s is $X = \frac{1}{k} \sum_i \sigma_i$. Therefore we have
	\[\EE{(Y_{[k]} - X)^2} = \EE{\parens{\frac{k - 1}{k} \sum_i \sigma_i}^2} = \parens{\frac{k - 1}{k}}^2 \EE{\parens{\sum_i \sigma_i}^2}.\]
	On the other hand, we have that $\EE{(Y_{[k]} - \EE{Y})^2} = \EE{(\sum_i \sigma_i)^2}$, so (using Corollary~\ref{cor:improvement_rewrite}) we have that the approximation ratio is
	\[1 - \parens{\frac{k - 1}{k}}^2 = \frac{2}{k} - \frac{1}{k^2}.\]
	This completes the first step. To complete the second step, we use the following well-known result from statistical theory.\footnote{This fact does not generalize to more than two dimensions, meaning that if the $x_i$ are vectors in three or more dimensions drawn independently from a normal distribution with unknown mean and known covariance matrix, then there is an estimator for the mean that Pareto dominates the sample mean according to expected squared vector distance. One such estimator is the \emph{James-Stein estimator} \cite{stein56}.}
	
	\begin{prop}[\cite{blyth51}, \cite{gs51}, \cite{hl51}] \label{prop:admissible}
		Let $x_1, \dots, x_k$ be drawn independently from a normal distribution with unknown mean $\mu$ and standard deviation $1$. Let $\hat{\mu} := \frac{1}{k} \sum_{i = 1}^k x_i$. Then for every function $X$ of $x_1, \dots, x_k$, there exists $\mu$ such that $\EE{(X - \mu)^2} \ge \EE{(\hat{\mu} - \mu)^2}$.
	\end{prop}
	
	Since $\EE{(\hat{\mu} - \mu)^2} = \frac{1}{k}$ for every $\mu$, we have the following fact as a corollary.
	
	\begin{corollary} \label{cor:mean_estimator}
		Let $x_1, \dots, x_k$ be drawn independently from a normal distribution with unknown mean $\mu$ and standard deviation $1$. Then for every aggregation strategy $X$ that takes as input $x_1, \dots, x_k$, we have $\max_\mu \EE{(X - \mu)^2} \ge \frac{1}{k}$.
	\end{corollary}
	
	(The only subtlety is that aggregation strategies are not required to be deterministic; however, replacing a randomized aggregation strategy with the deterministic strategy that outputs the expected value of the randomized strategy given the inputs can only reduce expected squared error.)\\
	
	Returning to our proof, let $X$ be any aggregation strategy on inputs $Y_1, \dots, Y_k$. We define a new aggregation strategy: $\tilde{X} := \frac{1}{k - 1}(\sum_i Y_i - X)$. We claim that if $X$ achieves an approximation ratio of more than $\frac{2}{k} - \frac{1}{k^2}$ on every $\mu$, then $\tilde{X}$ violates Corollary~\ref{cor:mean_estimator}.
	
	Note that $Y_i = \sigma_i + (k - 1)\mu$, which means that the $Y_i$'s are independent samples from $N(k\mu, 1)$. Consider $\tilde{X}$ as an estimator for $k\mu$. We have
	
	\begin{align*}
		\EE{(k\mu - \tilde{X})^2} &= \EE{\parens{k\mu - \frac{1}{k - 1}\parens{\sum_i Y_i - X}}^2}\\
		&= \frac{1}{(k - 1)^2} \EE{\parens{\sum_i Y_i - k(k - 1)\mu - X}^2} = \frac{1}{(k - 1)^2} \EE{(Y - X)^2},
	\end{align*}
	where in the last step we use the fact that $Y = \sum_i \sigma_i = \sum_i Y_i - k(k - 1)\mu$. Now, suppose for contradiction that $X$ achieves an approximation ratio of more than $\frac{2}{k} - \frac{1}{k^2}$ on every $\mu$. Then for all $\mu$ we have
	\[\parens{\frac{k - 1}{k}}^2 = 1 - \parens{\frac{2}{k} - \frac{1}{k^2}} > \frac{\EE{(Y - X)^2}}{\EE{(Y - \EE{Y})^2}} = \frac{\EE{(Y - X)^2}}{k} = \frac{(k - 1)^2}{k} \EE{(k\mu - \tilde{X})^2},\]
	so $\EE{(k\mu - \tilde{X})^2} < \frac{1}{k}$ for every value of $k\mu$ (and thus for every $\mu$). This contradicts Corollary~\ref{cor:mean_estimator} and completes the proof.
\end{proof}

Theorems~\ref{thm:prior_free_positive} and \ref{thm:prior_free_negative} give us non-matching lower and upper bounds on the optimal approximation ratio under the projective substitutes condition. In particular, for $k = 2$ experts, Theorem~\ref{thm:prior_free_positive} tells us that averaging achieves an approximation ratio of $\frac{3 + \sqrt{7}}{8} \approx 0.706$, while Theorem~\ref{thm:prior_free_positive} tells us that no aggregation strategy can achieve an approximation ratio larger than $0.75$. We now show that for two experts, our positive result is tight.

\begin{theorem} \label{thm:prior_free_negative_n2}
	In the prior-free setting, no aggregation strategy achieves an approximation ratio larger than $\frac{3 + \sqrt{7}}{8}$ on every two-expert information structure that satisfies projective substitutes.
\end{theorem}

\begin{proof}
	Let $\mathcal{I}_+$ be the following information structure, where $p = 1 - \frac{\sqrt{7}}{4}$ and $x = \sqrt{14} - 2\sqrt{2}$. We label the signals $-1$ and $1$ because these are the expected values conditional on the respective signals.
	\[\mathcal{I}_+ := \left\{Y = \begin{tabular}{c|cc}
		&$\sigma_2 = 1$&$\sigma_2 = -1$\\
		\hline
		$\sigma_1 = 1$&$\frac{1 - (1 - 2p)x}{2p}$&$x$\\
		$\sigma_1 = -1$&$x$&$\frac{-1 - (1 - 2p)x}{2p}$
	\end{tabular}
	\qquad \PP = \begin{tabular}{c|cc}
		&$\sigma_2 = 1$&$\sigma_2 = -1$\\
		\hline
		$\sigma_1 = 1$&$p$&$\frac{1}{2} - p$\\
		$\sigma_1 = -1$&$\frac{1}{2} - p$&$p$
	\end{tabular}\right\}
	\]
	Let $\mathcal{I}_-$ be the same information structure, but with $x = 2\sqrt{2} - \sqrt{14}$. It is a matter of calculation to verify that these information structures satisfy projective substitutes.\footnote{These information structures were found by finding values that would make the inequalities in the proofs of Theorem~\ref{thm:prior_free_positive} and Lemma~\ref{lem:ab} hold with equality.}
	
	Note that any aggregation strategy that outputs a number other than $1$ on input $(Y_1, Y_2) = (1, 1)$ has an approximation ratio of negative infinity on an information structure where $Y = 1$ deterministically. This is likewise true for $-1$ in place of $1$. Thus, if Theorem~\ref{thm:prior_free_negative_n2} were false, it would be disproved by an information structure that outputs $1$ on $\mathcal{I}_+$ and $\mathcal{I}_-$ if $(\sigma_1, \sigma_2) = (1, 1)$ and $-1$ if $(\sigma_1, \sigma_2) = (-1, -1)$. Conditional on this, the aggregation method that minimizes the maximum expected squared distance to $Y_{[k]}$ on $\mathcal{I}_+$ and $\mathcal{I}_-$ is the one that returns $0$ when $(\sigma_1, \sigma_2) = (1, -1)$ or $(\sigma_1, \sigma_2) = (-1, 1)$. It is a matter of calculation to verify that this aggregation strategy achieves an approximation ratio of exactly $\frac{3 + \sqrt{7}}{8}$.
\end{proof}

%% file: Known_Prior.tex
\section{The Known Prior Setting} \label{sec:known_prior}
Let us now expand the information available to the aggregator by allowing them knowledge of the prior $\EE{Y}$. How might this change the optimal aggregation strategy?

Consider the following information structure: a coin comes up heads $Y$ fraction of the time, selected uniformly from $[0, 1]$. Each of two experts sees an independent flip of the coin. It can be calculated that an expert who sees heads has a posterior of $\frac{2}{3}$. However, consider the situation in which both experts report heads: collectively they have seen two heads and zero tails, conditional on which the expected value of $Y$ is $\frac{3}{4}$.

In this setting, it is beneficial to push the average of the experts' reports away from the prior, a method known as \emph{extremization} in the forecasting literature. Extremization has been demonstrated empirically to improve the accuracy of aggregate forecasts, and the theoretical intuition for the benefit of extremization extends beyond the above example. An expert's report is the posterior resulting from observing a fraction of all available evidence. If many experts update upward from the prior as a result of each of their pieces of evidence, then it stands to reason that observing \emph{all of the evidence} would cause an update that is larger than the update of an average expert.

\subsection{The Extremization Factor}
Consider the following aggregation strategy, parameterized by a constant $d$ which we will call the \emph{extremization factor}.
\begin{equation} \label{eq:ext_factor}
	X := \frac{1}{k} \sum_i Y_i + (d - 1) \parens{\frac{1}{k} \sum_i Y_i - \EE{Y}}.
\end{equation}
Setting $d = 1$ recovers the average of the reports; setting $d = 0$ simply returns the prior. In general, setting $d > 1$ extremizes the average (i.e. pushes it away from the prior) by a factor of $d$. As an example, consider the class of information structures in Theorem~\ref{thm:prior_free_negative}, where averaging achieved an approximation ratio of $\frac{2}{k} - \frac{1}{k^2}$. On the other hand, extremizing by a factor of $k$ (i.e. setting $d = k$ above) recovers $Y$ exactly (thus achieving an approximation ratio of $1$).

We now prove that by extremizing, we can achieve an approximation ratio that is higher than what we could hope to attain without knowledge of the prior. In particular, we find that \textbf{by extremizing, it is possible to achieve an approximation ratio of at least $\frac{3\sqrt{3}}{2k} - O(1/k^2) \approx 2.598/k$}, a substantial improvement not only over our positive result in the prior-free setting, but also over our negative result.

\subsection{Positive Result for the Known Prior Setting}
\begin{theorem} \label{thm:known_prior_positive}
	Let $\mathcal{I} = (\Omega, \PP, S, Y)$ be an information structure for $k$ experts that satisfies projective substitutes, and let $X = \frac{1}{k} \sum_{i = 1}^k Y_i + (d - 1) \parens{\frac{1}{k} \sum_{i = 1}^n Y_i - \EE{Y}}$, where $d = \frac{k(\sqrt{3k^2 - 3k + 1} - 2)}{k^2 - k - 1}$. Then $X$ attains an approximation ratio of at least
	\[\frac{(3k^2 - 3k + 1)^{3/2} - 9k^2 + 9k + 1}{2(k^2 - k - 1)^2} \ge \frac{3\sqrt{3}}{2k} - O \parens{\frac{1}{k^2}}.\]
\end{theorem}

The proof is similar to the proof of Theorem~\ref{thm:prior_free_positive}. We provide a sketch of the proof and refer to Appendix~\ref{appx:prelims_omitted} for the details.

\begin{proof}[Proof sketch]
	Let $d > 0$ and let $X := \frac{1}{k} \sum_i Y_i + (d - 1) \parens{\frac{1}{k} \sum_i Y_i - \EE{Y}}$. By rearranging terms and applying Equation~\ref{eq:formula_one}, we find that
	\[\EE{(Y_{[k]} - X)^2} = (d - 1)^2 \EE{(Y_{[k]} - \EE{Y})^2} - \frac{d(d - 2)}{k} \sum_i \EE{(Y_{[k]} - Y_i)^2} - \frac{d^2}{k^2} \sum_{1 \le i < j \le k} \EE{(Y_i - Y_j)^2}.\]
	We then apply Lemma~\ref{lem:ab} and use the Pythagorean theorem and Equation~\ref{eq:conditional_step} to find that for any $a, b \ge 0$ such that $b \ge \frac{(2a - 1)^2}{4a}$, if $\frac{d(d - 2)}{k} + \frac{bd^2(k - 1)}{k^2} + \frac{ad^2}{k^2} \le 0$ then
	\[\EE{(Y_{[k]} - X)^2} \le \parens{1 + \frac{d^2 - 2d}{k} - (a - b) \frac{k - 1}{k^2} d^2} \EE{(Y_{[k]} - \EE{Y})^2}.\]
	Just as in the proof of Theorem~\ref{thm:prior_free_positive}, our goal is now to maximize $a - b$, this time with the constraints $b \ge \frac{(2a - 1)^2}{4a}$ and $\frac{d(d - 2)}{k} + \frac{bd^2(k - 1)}{k^2} + \frac{ad^2}{k^2} \le 0$. Solving this optimization problem yields the theorem statement.
\end{proof}

It is worth noting how $d$ varies with $k$. While $d$ increases with
$k$, it reaches a limit --- namely, $\sqrt{3} \approx
1.732$.\footnote{Interestingly, the value of~$d$ recommended by our theoretical
  analysis is consonant with the empirical findings in~\cite{su15}.}
  Thus,
the strategy given in Theorem~\ref{thm:known_prior_positive} does not
extremize in the way that one might guess based on the optimal
response to the information structures in our negative result for the
prior-free setting, Theorem~\ref{thm:prior_free_negative}.
This makes sense, because the signals received by each expert were independent; the worst-case optimal strategy, on the other hand, must compromise between doing well in such settings (where a large extremization factor makes sense) and ones where experts' information is highly dependent (where little or no extremization is optimal).

\subsection{Negative Results for the Known Prior Setting}
We first prove a conditional negative result: specifically, that our above approach of averaging and extremizing by a constant factor cannot achieve an approximation ratio better than $\frac{4}{k}$.

\begin{theorem} \label{thm:known_prior_conditional_negative}
	Fix any $k \ge 1$. Let $\mathcal{I}_{\text{ind}}$ be the information structure in which each expert $i$ receives an independent signal $\sigma_i \sim N(0, 1)$. Let $\mathcal{I}_{\text{eq}}$ be the information structure in which each expert $i$ receives the \emph{same} signal $\sigma_i$, which is also drawn from $\sim N(0, 1)$. For both information structures, $Y = \sum_{i = 1}^k \sigma_i$. Then
	
	\begin{enumerate}[label=(\roman*)]
		\item $\mathcal{I}_{\text{ind}}$ and $\mathcal{I}_{\text{dep}}$ satisfy projective substitutes.
		\item There is no $d \in \RR$ for which the aggregation strategy that averages the experts' reports and extremizes by a factor of $d$ attains an approximation ratio of more than $\frac{4k}{(k + 1)^2}$ on both $\mathcal{I}_{\text{ind}}$ and $\mathcal{I}_{\text{dep}}$.
	\end{enumerate}
\end{theorem}

As with Theorem~\ref{thm:prior_free_negative}, the result extends to all settings that permit both $\mathcal{I}_{\text{ind}}$ and $\mathcal{I}_{\text{dep}}$.

\begin{proof}
	The fact that $\mathcal{I}_{\text{ind}}$ satisfies projective substitutes follows from part \ref{item:negative_1} of Theorem~\ref{thm:prior_free_negative}. The fact that $\mathcal{I}_{\text{dep}}$ satisfies projective substitutes is clear because the right-hand side of Equation~\ref{eq:proj_subs} is zero.
	
	To show the second part, note that for $\mathcal{I}_{\text{ind}}$ we have $Y_i = \sigma_i$, so $Y = \sum_i Y_i$, and for $\mathcal{I}_{\text{dep}}$ we have $Y_i = k\sigma_i$, so $Y = \frac{1}{k} \sum_i Y_i$. Therefore, an extremization factor of $d$ (as in Equation~\ref{eq:ext_factor}) attains an approximation ratio of $1 - \frac{(k - d)^2}{k^2}$ for $\mathcal{I}_{\text{ind}}$ and $1 - (1 - d)^2$ for $\mathcal{I}_{\text{dep}}$. The $d$ that maximizes the minimum of these two expressions is $d = \frac{2k}{k + 1}$, which yields an approximation ratio of $\frac{4k}{(k + 1)^2}$ for both information structures.
\end{proof}

Finally, we note that Theorem~\ref{thm:known_prior_positive} tells us that by averaging and extremizing by a factor of $2(\sqrt{7} - 2) \approx 1.292$ achieves an approximation ratio of $\frac{7\sqrt{7} - 17}{2} \approx 0.760$. We prove that this is tight. Note that this result, unlike Theorem~\ref{thm:known_prior_conditional_negative}, is unconditional.

\begin{theorem} \label{thm:known_prior_negative_n2}
	In the known prior setting, no aggregation strategy achieves an approximation ratio larger than $\frac{7\sqrt{7} - 17}{2}$ on every two-expert information structure that satisfies projective substitutes.
\end{theorem}

\begin{proof}
Let $\mathcal{I}_+$ be the following information structure, where $p = \frac{2 + \sqrt{7}}{12}$ and $x = \frac{\sqrt{2 + \sqrt{7}}}{3}$. We label the signals $-1$ and $1$ because these are the expected values conditional on the respective signals.
\[\mathcal{I}_+ := \left\{Y = \begin{tabular}{c|cc}
	&$\sigma_2 = 1$&$\sigma_2 = -1$\\
	\hline
	$\sigma_1 = 1$&$\frac{1 - (1 - 2p)x}{2p}$&$x$\\
	$\sigma_1 = -1$&$x$&$\frac{-1 - (1 - 2p)x}{2p}$
\end{tabular}
\qquad \PP = \begin{tabular}{c|cc}
	&$\sigma_2 = 1$&$\sigma_2 = -1$\\
	\hline
	$\sigma_1 = 1$&$p$&$\frac{1}{2} - p$\\
	$\sigma_1 = -1$&$\frac{1}{2} - p$&$p$
\end{tabular}\right\}
\]
Let $\mathcal{I}_-$ be the same information structure, but with $x = -\frac{\sqrt{2 + \sqrt{7}}}{3}$. It is a matter of calculation to verify that these information structures satisfy projective substitutes. The quantity $\EE{(Y - \EE{Y})^2}$ is the same for $\mathcal{I}_+$ and $\mathcal{I}_-$, so the aggregation strategy $X$ that guarantees the largest possible approximation ratio when the information structure is one of $\mathcal{I}_+$ and $\mathcal{I}_-$ is the one that minimizes the maximum value of $\EE{(Y - X)^2}$ over these two information structures. This is achieved by outputting $0$ when $(Y_1, Y_2)$ is $(1, -1)$ or $(-1, 1)$, $\frac{1}{2p}$ when $(Y_1, Y_2) = (1, 1)$, and $\frac{-1}{2p}$ when $(Y_1, Y_2) = (-1, -1)$. It is a matter of calculation to verify that this aggregation strategy achieves an approximation ratio of exactly $\frac{7\sqrt{7} - 17}{2}$.
\end{proof}

\begin{remark} \label{rem:other_information}
	The example in Theorem~\ref{thm:known_prior_negative_n2} in fact shows that for $k = 2$ experts it is impossible to achieve an approximation ratio larger than $\frac{7\sqrt{7} - 17}{2}$ even if the aggregator has access to substantially more information than the prior. In particular, the aggregator cannot do better even if given access to $\PP$ and to the random variables $Y_1, Y_2$ as functions (i.e. what the value of $Y_1$ is under every possible signal outcome $\sigma_1$, and similarly for $Y_2$).
\end{remark}

\subsection{Beyond Linear Extremization?}
In this section we have focused on what \cite{su15} call \emph{linear extremization}: extremization by a fixed multiplicative factor. To our knowledge this is the only extremization method for real-valued forecasts that has been studied.\footnote{Other methods have been studied for extremizing probabilistic forecasts, see e.g. \cite{sbfmtu14}, but that is not our setting.} As byproducts of our upper and lower bounds on the power of this technique, our work provably separates what is possible with and without knowledge of the prior, justifies in a formal sense the practice of linear extremization (by showing that it is superior to straightforward averaging), and provides guidance as to how much extremization is appropriate.

Our general  model reveals a natural open question, namely whether there are more exotic aggregation strategies guaranteed to outperform linear extremization. This question, which can easily be phrased formally using the definitions and performance metrics of our general model, appears not to be formally explored in the existing literature and is an intriguing question for future work.

%% file: Appendix.tex
\section{Omitted Proofs} \label{appx:prelims_omitted}
\begin{proof}[Proof of Proposition~\ref{prop:weak_subs_rewrite}]
	Note that for an information structure $\mathcal{I} = (\Omega, \PP, S, Y)$ for $k$ experts and subsets $A \subseteq B \subseteq [k]$, we have
	\begin{equation} \label{eq:v_diff}
		v(Y_B) - v(Y_A) = \EE{(Y_B - Y_A)^2}.
	\end{equation}
	\[\]
	Indeed, we have
	\begin{align*}
		v(Y_B) - v(A) &= (\EE{(Y - \EE{Y})^2} - \EE{(Y - Y_B)^2}) - (\EE{(Y - \EE{Y})^2} - \EE{(Y - Y_A)^2})\\
		&= \EE{(Y - Y_A)^2} - \EE{(Y - Y_B)^2} = \EE{(Y_B - Y_A)^2},
	\end{align*}
	where the last step follows by setting the variables $A, B, C$ in Proposition~\ref{prop:pythag} to $Y$, $Y_B$, and $Y_A$, respectively. Proposition~\ref{prop:weak_subs_rewrite} follows directly from Equation~\ref{eq:v_diff} after rearranging the terms in Equation~\ref{eq:weak_subs}.
\end{proof}

\begin{proof}[Proof of Proposition~\ref{prop:pythag}]
	Observe that $\EE{AB} = \EE{\EE{AB \mid B}} = \EE{B^2}$, so $\EE{(A - B)^2} = \EE{A^2} - \EE{B^2}$. Additionally, note that
	\[\EE{AC} = \EE{\EE{AC \mid \mathcal{F}}} = \EE{\EE{A \mid \mathcal{F}} C} = \EE{BC},\]
	where in the second step we use that $C$ is defined on $\mathcal{F}$. Therefore, we have
	\begin{align*}
		\EE{(A - C)^2} &= \EE{A^2} + \EE{C^2} - 2\EE{BC} = \EE{A^2} - \EE{B^2} + \EE{B^2} - 2\EE{BC} + \EE{C^2}\\
		&= \EE{(A - B)^2} + \EE{(B - C)^2},
	\end{align*}
	as desired.
\end{proof}

\begin{proof}[Proof of Corollary~\ref{cor:improvement_rewrite}]
	We have
	\begin{align*}
		\frac{v(Z)}{v(Y_{[k]})} &= \frac{\EE{(Y - \EE{Y})^2} - \EE{(Y - Z)^2}}{\EE{(Y - \EE{Y})^2} - \EE{(Y - Y_{[k]})^2}}\\
		&= \frac{\parens{\EE{(Y - Y_{[k]})^2} + \EE{(Y_{[k]} - \EE{Y})^2}} - \parens{\EE{(Y - Y_{[k]})^2} + \EE{(Y_{[k]} - Z)^2}}}{\EE{(Y - \EE{Y})^2} - \EE{(Y - Y_{[k]})^2}}\\
		&= \frac{\EE{(Y_{[k]} - \EE{Y})^2} - \EE{(Y_{[k]} - Z)^2}}{\EE{(Y - Y_{[k]})^2}} = 1 - \frac{\EE{(Y_{[k]} - Z)^2}}{\EE{(Y_{[k]} - \EE{Y})^2}}.
	\end{align*}
	We use Proposition~\ref{prop:pythag} for $A = Y$, $B = Y_{[k]}$, $C = \EE{Y}$ in the second and third steps, and also for $A = Y$, $B = Y_{[k]}$, $C = Z$ in the second step.
\end{proof}

\begin{proof}[Proof of Theorem~\ref{thm:known_prior_positive}]
Let $d > 0$ and let $X := \frac{1}{k} \sum_i Y_i + (d - 1) \parens{\frac{1}{k} \sum_i Y_i - \EE{Y}}$. As with the proof of Theorem~\ref{thm:prior_free_positive}, we start by upper bounding $\EE{(Y_{[k]} - X)^2}$. We have
\begin{align*}
	\EE{(Y_{[k]} - X)^2} &= \EE{\parens{d \parens{Y_{[k]} - \frac{1}{k} \sum_i Y_i} - (d - 1) \parens{Y_{[k]} - \EE{Y}}}^2}\\
	&= d^2 \EE{\parens{Y_{[k]} - \frac{1}{k} \sum_i Y_i}^2} + (d - 1)^2 \EE{(Y_{[k]} - \EE{Y})^2}\\
	&\qquad - \frac{2d(d - 1)}{k} \EE{\parens{k Y_{[k]} - \sum_i Y_i}\parens{Y_{[k]} - \EE{Y}}}\\
	&= d^2 \parens{\frac{1}{k} \sum_i \EE{(Y_{[k]} - Y_i)^2} - \frac{1}{k^2} \sum_{1 \le i < j \le k} \EE{(Y_i - Y_j)^2}}\\
	&\qquad + (d - 1)^2 \EE{(Y_{[k]} - \EE{Y})^2} - \frac{2d(d - 1)}{k} \sum_i \EE{(Y_{[k]} - Y_i)^2}.
\end{align*}
In the last step, we adapt the first term using Equation~\ref{eq:formula_one} and adapt the last term by observing that
\begin{align*}
	\EE{(Y_{[k]} - Y_i)(Y_{[k]} - \EE{Y})} &= \EE{Y_{[k]}(Y_{[k]} - Y_i)} - \EE{Y}\EE{Y_{[k]} - Y_i}\\
	&= \EE{Y_{[k]}(Y_{[k]} - Y_i)} = \EE{(Y_{[k]} - Y_i)^2}
\end{align*}
(where the last step holds because for any given $Y_i$, $\EE{Y_{[k]} \mid Y_i} = Y_i$, so $\EE{Y_i(Y_{[k]} - Y_i)} = 0$). Grouping like terms, we have
\[\EE{(Y_{[k]} - X)^2} = (d - 1)^2 \EE{(Y_{[k]} - \EE{Y})^2} - \frac{d(d - 2)}{k} \sum_i \EE{(Y_{[k]} - Y_i)^2} - \frac{d^2}{k^2} \sum_{1 \le i < j \le k} \EE{(Y_i - Y_j)^2}.\]
Now, recall Lemma~\ref{lem:ab}. Consider any $a, b \ge 0$ satisfying $b \ge \frac{(2a - 1)^2}{4a}$; then for all $i, j$ we have
\[\EE{(Y_i - Y_j)^2} \ge a \parens{\EE{(Y_{\{i,j\}} - Y_i)^2} + \EE{(Y_{\{i,j\}} - Y_j)^2}} - b \parens{\EE{(Y_i - \EE{Y})^2} + \EE{(Y_j - \EE{Y})^2}}.\]
Therefore we have
\begin{align*}
	&\EE{(Y_{[k]} - X)^2} \le (d - 1)^2 \EE{(Y_{[k]} - \EE{Y})^2} - \frac{d(d - 2)}{k} \sum_i \EE{(Y_{[k]} - Y_i)^2}\\
	&\qquad - \frac{d^2}{k^2} \sum_{1 \le i < j \le k} \parens{a \parens{\EE{(Y_{\{i,j\}} - Y_i)^2} + \EE{(Y_{\{i,j\}} - Y_j)^2}} - b \parens{\EE{(Y_i - \EE{Y})^2} + \EE{(Y_j - \EE{Y})^2}}}\\
	&= \parens{(d - 1)^2 + \frac{bd^2(k - 1)}{k}} \EE{(Y_{[k]} - \EE{Y})^2} - \parens{\frac{d(d - 2)}{k} + \frac{bd^2(k - 1)}{k^2}} \sum_i \EE{(Y_{[k]} - Y_i)^2}\\
	&\qquad - \frac{ad^2}{k^2} \sum_{1 \le i < j \le k} \parens{\EE{(Y_{\{i,j\}} - Y_i)^2} + \EE{(Y_{\{i,j\}} - Y_j)^2}},
\end{align*}
where in the last step we use the Pythagorean theorem to write $\EE{(Y_i - \EE{Y})^2}$ as $\EE{(Y - \EE{Y})^2} - \EE{(Y - Y_i)^2}$. Now we use Equation~\ref{eq:weak_subs_use}:
\begin{align*}
	\EE{(Y_{[k]} - X)^2} &\le \parens{(d - 1)^2 + \frac{bd^2(k - 1)}{k}} \EE{(Y_{[k]} - \EE{Y})^2}\\
	&\qquad - \parens{\frac{d(d - 2)}{k} + \frac{bd^2(k - 1)}{k^2} + \frac{ad^2}{k^2}} \sum_i \EE{(Y_{[k]} - Y_i)^2}.
\end{align*}
Now, supposing that $\frac{d(d - 2)}{k} + \frac{bd^2(k - 1)}{k^2} + \frac{ad^2}{k^2}$ is not positive, we may use Equation~\ref{eq:conditional_step} to obtain:
\begin{align*}
	\EE{(Y_{[k]} - X)^2} &\le \parens{(d - 1)^2 + \frac{bd^2(k - 1)}{k} - \frac{k - 1}{k} \parens{d(d - 2) + \frac{bd^2(k - 1)}{k} + \frac{ad^2}{k}}} \EE{(Y_{[k]} - \EE{Y})^2}\\
	&= \parens{1 + \frac{d^2 - 2d}{k} - (a - b) \frac{k - 1}{k^2} d^2} \EE{(Y_{[k]} - \EE{Y})^2}.
\end{align*}
With $d$ held fixed, our goal is to maximize $a - b$, just as in the proof of Theorem~\ref{thm:prior_free_positive}. This time, our constraints are $b \ge \frac{(2a - 1)^2}{4a}$ (as before) and $\frac{d(d - 2)}{k} + \frac{bd^2(k - 1)}{k^2} + \frac{ad^2}{k^2} \le 0$, which can be rewritten as $a + b(k - 1) \le \frac{2 - d}{d} k$. The optimal values are
\[a = \frac{\frac{2}{d}k - 1 + \sqrt{\parens{\frac{2}{d}k - 1}^2 - k(k - 1)}}{2k} \text{ and } b = \frac{(2a - 1)^2}{4a}.\]
Now, let $a$ and $b$ be as above. We may select $d$ as we please and seek to minimize the expression
\[1 + \frac{d^2 - 2d}{k} - (a - b) \frac{k - 1}{k^2} d^2.\]
We choose the value of $d$ in the theorem statement (which one can verify is optimal using a computer algebra system). This yields the desired approximation ratio.
\end{proof}